\documentclass[pteplogo]{ptephy_v1}

\usepackage{hyperref}

\usepackage{amsmath,mathrsfs,amssymb,amsthm}
\usepackage{braket}
\usepackage{fancyhdr}

\usepackage{mathtools}
\usepackage{here}
\usepackage{color}
\usepackage{graphicx}
\usepackage{enumerate}
\usepackage{tensor}
\usepackage{ascmac}
\usepackage[normalem]{ulem}
\newtheorem{proposition}{Proposition}

\usepackage{comment}

\newcommand{\nn}{\nonumber} 
\newcommand{\mo}[2]{{\mathcal O}(#1^{#2})} 
\newcommand{\dd}{\mathrm{d}}   
\newcommand{\p}[1]{\partial_{#1}}   

\begin{document}

\title{Modification of a Lie algebra-based approach and its application to asymptotic symmetries on a Killing horizon}


\author{Takeshi Tomitsuka}
\affil{Graduate School of Science, Tohoku University, Sendai, 980-8578, Japan }

\author[2]{Koji Yamaguchi}
\affil{Department of Applied Mathematics, University of Waterloo, Waterloo, Ontario, N2L 3G1, Canada}




\begin{abstract}%
We develop a new approach to find asymptotic symmetries in general relativity as a modification of a Lie algebra based approach proposed in Ref.~\cite{tomitsuka2021}. The authors in Ref.~\cite{tomitsuka2021} proposed an algorithmic protocol to investigate asymptotic symmetries. In particular, their guiding principle helps us to find a non-vanishing charge which generates an infinitesimal diffeomorphism. 
However, in order to check the integrability condition for the charges, it is required to solve differential equations to identify the integral curve of vector fields, which is usually quite hard. 
In this paper, we provide a sufficient condition of the integrability condition which can be checked without solving any differential equations, avoiding difficulties in the approach in Ref.~\cite{tomitsuka2021}.
As a demonstration, we investigate the asymptotic symmetries on a Killing horizon and find a new class of asymptotic symmetries.
In 4-dimensional spacetimes with a spherical Killing horizon, we show that the algebra of the corresponding charges is a central extension of the algebra of vector fields. 
\end{abstract}


\maketitle

\section{Introduction}
\label{sec:intro}
The uniqueness theorem \cite{Israel1967,Israel1968,Carter1971} states that every 4-dimensional stationary black hole solution to the Einstein-Maxwell equations in general relativity is completely characterized by just three parameters, mass, angular momentum, and electric charge. 
On the other hand, Bekenstein \cite{Bekenstein1972} proposed that a black hole has entropy proportional to its horizon area $A$, and Hawking \cite{hawking_1974_black_hole_explosions} showed that a black hole emits the thermal radiation and has what we call the Bekenstein-Hawking (BH) entropy $A/4G$, where $G$ denotes the gravitational constant. 
It suggests that the black hole has a lot of microstates even though it can be characterized by the above three parameters.
What is the origin of such microstates?

So far, a great deal of effort has been devoted to explaining the origin of the BH entropy. 
One possible origin is the so-called asymptotic symmetries on a horizon. 
General relativity is invariant under diffeomorphisms.
Sometimes, it is argued that diffeomorphisms are gauge transformations in general relativity, which do not change the state of the system physically. 
If so, the metrics connected by diffeomorphisms cannot be distinguished from each other and hence diffeomorphisms may seem to have nothing to do with the origin of microstates.

However, in fact, not all diffeomorphisms generate gauge 
transformations. 
A way to judge whether a diffeomorphism is not a gauge transformation is to check the value of the charge generating the transformation. 
If the value of a charge is modified by a diffeomorphism, then it is not a gauge transformation since the original metric and transformed one can be discriminated. 
Such a physical 
transformation generates microstates that may contribute to the BH 
entropy.
As well known, the value of a charge generating an infinitesimal diffeomorphism is given by an integral over the boundary of a spacetime in general relativity.
Thus, the asymptotic behaviors of a diffeomorphism and the metric play a crucial role to identify transformations which cannot be gauged away. 
Such asymptotic symmetries of spacetimes have been investigated as a possible origin of the BH entropy, e.g. in Refs.~\cite{Carlip_99,Hotta_2001,Hotta_2002,Hawking_2016,Grumiller_2020,PhysRevLett.125.241302}.

Despite such an importance, studies on asymptotic symmetries often take enormous efforts.
In the conventional approach, we first specify the asymptotic
behavior of the metrics near the boundary and solve the 
asymptotic Killing equation.
The set of all asymptotic Killing vectors forms an algebra which generates a diffeomorphism. 
Next, we check whether the so-called integrability condition is satisfied. 
If it is not, we have to go back to the beginning to get a well-defined charges. 
Even when the integrability condition is satisfied, there remains a 
possibility that all the charges vanish for any metrics in question.
In this case, since the metrics cannot be discriminated by the values of the charges, the diffeomorphism can be gauged away.
Thus, to find non-trivial charges, we also have to restart the above protocol from the beginning. 
In this sense, it is important but sometimes difficult to find an appropriate asymptotic behavior of the metrics in the first step which result in non-trivial and integrable charges by trials and errors. 

As an alternative approach, the authors of this paper and a collaborator proposed the Lie algebra based 
approach in Ref.~\cite{tomitsuka2021}. 
In contrast to the conventional approach, in our approach, we first pick up a pair of two vector fields such that the Poisson bracket of the charges generating infinitesimal diffeomorphisms along them does not vanish at a fixed but arbitrary metric $\bar{g}_{\mu\nu}$, which we call the background metric. 
We then fix a Lie algebra $\mathcal{A}$, which contains those vector fields. 
Instead of the metrics with an asymptotic behavior introduced by hand, we adopted the set of metrics $\mathcal{S}$ which are connected to $\bar{g}_{\mu\nu}$ by diffeomorphisms generated by $\mathcal{A}$.
The algebra of the charges is non-trivial by construction as long as the integrability condition is satisfied since there is a set of elements whose Poisson bracket does not vanish. 
In Ref.~\cite{tomitsuka2021}, we applied this approach to the Rindler horizon and found the new symmetry which we call superdilatations.

Although our approach proposed in Ref.~\cite{tomitsuka2021} may be powerful in finding asymptotic symmetries, there remain hard tasks which are required to check the integrability of the charges directly. 
We need to solve differential equations to obtain all the diffeomorphisms generated by $\mathcal{A}$ and identify $\mathcal{S}$. 
Although there are examples of algebras $\mathcal{A}$ for which the differential equations can be solved, e.g., those in Ref.~\cite{tomitsuka2021}, in general, it is quite difficult to solve the differential equations for a given $\mathcal{A}$.
In this paper, we propose a modification of the approach to overcome this issue. 
A key ingredient is a sufficient condition for charges to be integrable, which can be checked at the background metric $\bar{g}_{\mu\nu}$.
It enables us to check the integrability condition without solving any differential equation. 
Since the algebra of integrable charges can be fully characterized by calculating the value of the Poisson bracket at the background metric $\bar{g}_{\mu\nu}$, there is no need to identify diffeomorphisms generated by $\mathcal{A}$ or $\mathcal{S}$ directly.
As an explicit example, we investigate the asymptotic symmetries on the Killing horizon with our approach. 
We find a new asymptotic symmetry composed of a class of supertranslations, superrotations and superdilatations in $D$-dimensional spacetimes with the Killing horizon. 
In particular, the algebra of the charges in 4-dimensional spacetimes with a spherical Killing horizon is calculated explicitly, which is shown to be a central extension of $\mathcal{A}$. 

This paper is organized as follows:
In Sec.~\ref{sec:Wald_method}, we briefly review the covariant phase space method, which is adopted in this paper to construct the charges generating infinitesimal diffeomorphisms.
In Sec.~\ref{sec:review_on_Lie}, we briefly review a Lie algebra based approach proposed in Ref.~\cite{tomitsuka2021}.
In Sec.~\ref{sec:integrability}, we provide a sufficient condition for the charges to be integrable.
In Sec.~\ref{sec:Killing}, we find a new symmetry on the Killing horizon by using our approach and investigate the algebra of its charges.
In Sec.~\ref{sec:summary}, we present the summary of this paper.
In this paper, we set the speed of light to unity: $c=1$.

\section{Covariant phase space method}\label{sec:Wald_method}
Let us briefly review the covariant phase space method developed in \cite{1987thyg.book..676C,Crnkovic_1988,Lee_Wald_1990,Wald_1993,Iyer_Wald_1994,Iyer_Wald_1995,Wald_Zoupas_2000}, which will be adopted in the rest of this paper. 
This method enables us to investigate and construct the algebra of the charges in an independent way of the choice of a local coordinate system.

We here focus on the gravitational system without any matter field. 
The Einstein-Hilbert action with the cosmological constant $\Lambda$ is given by
\begin{align}
	S & = \int_{\mathcal M}\dd^{D}x {\mathcal L}_{EH},\quad {\mathcal L}_{EH} \coloneqq \frac{\sqrt{-g}}{16\pi G}(R-2\Lambda),
\end{align}
where $\int_\mathcal{M}\dd^{D}x$ denotes the integral over a $D$-dimensional spacetime $\mathcal{M}$, $g$ and $R$ are the determinant of the metric $g_{\mu\nu}$ and the Ricci scalar, respectively.
The variation of the integrand is decomposed into two parts as
\begin{align}
	\delta {\mathcal L}_{EH} & =  -\frac{\sqrt{-g}}{16\pi G}(G^{\mu\nu}+\Lambda g^{\mu\nu})\delta g_{\mu\nu} + \partial_{\mu}\Theta^{\mu}(g,\delta g),\label{eq:variation_LEH}
\end{align}
where $G_{\mu\nu}$ is the Einstein tensor defined by
\begin{align}
    G_{\mu\nu} \coloneqq R_{\mu\nu} - \frac{1}{2}R g_{\mu\nu},
\end{align} 
while $\Theta$ is called the pre-symplectic potential, which is defined by
\begin{align}
	\Theta^{\mu}(g, \delta g) & = \frac{\sqrt{-g}}{16\pi G}\left(g^{\mu\alpha}\nabla^{\beta}\delta g_{\alpha\beta} - g^{\alpha\beta}\nabla^{\mu}\delta g_{\alpha\beta}\right).
\end{align}

For an infinitesimal transformation of metric $\delta_\xi g_{\mu\nu}\coloneqq \pounds_\xi g_{\mu\nu}$, where $\xi$ is a vector field and $\pounds_\xi$ denotes the Lie derivative along it, we have
\begin{align}
    \delta_\xi\mathcal {L}_{EH}=\pounds_\xi \mathcal {L}_{EH}=\partial_\mu \left(\xi^\mu\mathcal {L}_{EH}\right)\label{eq:lie_derivative_LEH}
\end{align}
since $\mathcal{L}_{EH}$ is a scalar density. 
Defining the Noether current
\begin{align}
	J^{\mu}[\xi] \coloneqq \Theta^{\mu}(g, \pounds_{\xi}g) -\xi^{\mu}{\mathcal L}_{EH},\label{eq:J_definition}
\end{align}
Eqs.~\eqref{eq:variation_LEH} and \eqref{eq:lie_derivative_LEH} imply that
\begin{align}
	\partial_\mu J^{\mu}[\xi]=\frac{\sqrt{-g}}{16\pi G} (G^{\mu\nu}+\Lambda g^{\mu\nu})\pounds_{\xi}g_{\mu\nu}
\end{align}
holds. From this equation, one can see that if $g_{\mu\nu}$ satisfies of the equation of motion, i.e., the Einstein equations $G_{\mu\nu}+\Lambda g_{\mu\nu}=0$, the current is conserved:
\begin{align}
	\partial_\mu J^\mu[\xi]\approx 0,
\end{align}
where $\approx$ denotes the equality which holds for any solution of the equation of motion.
In fact, we can decompose the current into two parts \cite{Iyer_Wald_1995}:
\begin{align}
    J^\mu[\xi]=\partial_\nu Q^{\mu\nu}[\xi]+\mathcal{C}\indices{^\mu_\nu}\xi^\nu,\quad 
    Q^{\mu\nu}[\xi]\coloneqq -\frac{\sqrt{-g}}{8\pi G}\nabla^{[\mu}\xi^{\nu]},\quad \mathcal{C}\indices{^\mu_\nu}\coloneqq \frac{\sqrt{-g}}{8\pi G}(G\indices{^\mu_\nu}+\Lambda g\indices{^\mu_\nu}).\label{eq:J_decomposition}
\end{align}
Here, the bracket $[\ ,\ ]$ for indices denotes an anti-symmetric symbol which is defined by
\begin{align}
	A_{[\mu_1\mu_2\cdots\mu_d]} \coloneqq \frac{1}{d!}\sum_{\sigma\in S_d}\mathrm{sgn}(\sigma)A_{\mu_{\sigma(1)}\mu_{\sigma(2)}\cdots\mu_{\sigma(d)}},
\end{align}
where $S_d$ is the permutation group and $\mathrm{sgn}(\sigma)$ denotes the signature of $\sigma\in S_d$. 

The Noether charge is defined as the integral of the current over a $(D-1)$-dinensional submanifold $\Sigma$ in $\mathcal{M}$ and given by
\begin{align}
	Q[\xi] \coloneqq \int_{\Sigma}(\dd^{D-1}x)_{\mu}J^{\mu}[\xi] 
	        \approx \int_{\Sigma}(\dd^{D-1}x)_{\mu}\partial_{\nu}Q^{\mu\nu}[\xi]  =\oint_{\partial\Sigma}(\dd^{D-2}x)_{\mu\nu}Q^{\mu\nu}[\xi],
	\label{Q_xi}
\end{align}
where $\partial\Sigma$ is the boundary of $\Sigma$. Note that the integral measure is given by
\begin{align}
	(\dd^{D-p}x)_{\mu_{1}\dots\mu_{p}} := \frac{\epsilon_{\mu_{1}\dots\mu_{p}\mu_{p+1}\dots\mu_{D}}}{p!(D-p)!}\dd x^{\mu_{p+1}}\wedge\dots\wedge\dd x^{\mu_{D}}, 
\end{align}
where $\epsilon_{\mu_{1}\dots\mu_{D}}$ is the $D$-dimensional Levi-Civita symbol.

For any linear perturbations of metric $\delta_1g_{\mu\nu}$ and $\delta_{2}g_{\mu\nu}$, the pre-symplectic current is defined as
\begin{align}
	\omega^{\mu}(g, \delta_{1}g, \delta_{2}g) := \delta_{1}\Theta^{\mu}(g, \delta_{2}g) - \delta_{2}\Theta^{\mu}(g,\delta_{1}g).
\end{align}
The integral of the pre-symplectic current over $\Sigma$ is called the pre-symplectic form, denoted by
\begin{align}
	\Omega(g, \delta_{1}g, \delta_{2}g) := \int_{\Sigma}(\dd^{D-1}x)_{\mu}\omega^{\mu}(g,\delta_{1}g, \delta_{2}g).
\end{align}
Let $H[\xi]$ denote the charge generating an infinitesimal transformation such that $g_{\mu\nu}\mapsto g_{\mu\nu}+\pounds_\xi g_{\mu\nu}$ for a vector field $\xi$. 
For a linear perturbation $\delta g_{\mu\nu}$, it is known \cite{Lee_Wald_1990,Wald_1993,Iyer_Wald_1994,Iyer_Wald_1995,Wald_Zoupas_2000} that the variation of the charge is given by 
\begin{align}
	\delta H[\xi] & = \Omega(g, \delta g, \pounds_{\xi}g) =\int_{\Sigma}(\dd^{D-1}x)_{\mu}\omega^{\mu}(g,\delta g, \pounds_{\xi}g).
\end{align}

From Eqs.~\eqref{eq:variation_LEH}, \eqref{eq:J_definition} and \eqref{eq:J_decomposition}, we get
\begin{align}
	\omega^\mu(g,\delta g, \pounds_\xi g)\approx\delta \mathcal{C}\indices{^\mu_\nu}\xi^\nu+\partial_\nu S^{\mu\nu}\left(g,\delta g,\pounds_\xi g\right),
	\label{onshell_pre}
\end{align}
where $S^{\mu\nu}\left(g,\delta g,\pounds_\xi g\right)$ is an anti-symmetric tensor defined by
\begin{align}
	S^{\mu\nu}\left(g,\delta g,\pounds_\xi g\right) & \coloneqq \delta Q^{\mu\nu}[\xi]+2\xi^{[\mu}\Theta^{\nu]}(g,\delta g)\nonumber \\
	&=\frac{\sqrt{-g}}{8\pi G}\Bigl(
	-\frac{1}{2}\delta g^{\alpha}_{\ \alpha}\nabla^{[\mu}\xi^{\nu]} + \delta g^{\alpha[\mu}\nabla_{\alpha}\xi^{\nu]} - \nabla^{[\mu}\delta g^{\nu]\alpha}\xi_{\alpha} + \xi^{[\mu}\nabla_{\alpha}\delta g^{\nu]\alpha}- \xi^{[\mu}\nabla^{\nu]}\delta g^{\alpha}_{\ \alpha}\Bigl).
	\label{eq_definition_S}
\end{align}
Therefore, if the generator $H[\xi]$ exists, its variation satisfies
\begin{align}
    \delta H[\xi]\approx \int_{\Sigma}\left(\dd^{D-1}x \right)_\mu\delta\mathcal{C}\indices{^\mu_\nu}\xi^\nu+\oint_{\partial\Sigma}\left(\dd^{D-2}x\right)_{\mu\nu}S^{\mu\nu}\left(g,\delta g,\pounds_\xi g\right). 
\end{align}
Assuming $\delta g_{\mu\nu}$ satisfies the linearized Einstein equations, the first term vanishes and we get
\begin{align}
    \delta H[\xi]\approx\oint_{\partial\Sigma}\left(\dd^{D-2}x\right)_{\mu\nu}S^{\mu\nu}\left(g,\delta g,\pounds_\xi g\right).\label{eq:variation_H}
\end{align}
This equation implies that the values of charges are characterized by the asymptotic behaviors of the metric $g_{\mu\nu}$, its perturbation $\delta g_{\mu\nu}$ and the vector field $\xi$. 

Let us now investigate under what condition the charge exists. 
Given the formula for the variation in Eq.~\eqref{eq:variation_H}, it must hold
\begin{align}
	  (\delta_{1}\delta_{2}-\delta_{2}\delta_{1})H[\xi]=0\label{eq:deltaH_commute}
\end{align}
for any variation $\delta_1$ and $\delta_2$ as the partial derivatives of a multivariable function commute. Since it holds
\begin{align}
    (\delta_{1}\delta_{2}-\delta_{2}\delta_{1})H[\xi]& = -\int_{\partial \Sigma}(\dd^{D-2}x)_{\mu\nu}\left(\xi^{[\mu}\delta_{1}\Theta^{\nu]}(g,\delta_{2}g)-\xi^{[\mu}\delta_{2}\Theta^{\nu]}(g,\delta_{1}g)\right) \nonumber \\
	  & =-\int_{\partial\Sigma}(\dd^{D-2}x)_{\mu\nu}\xi^{[\mu}\omega^{\nu]}(g,\delta_{1}g,\delta_{2}g)\nonumber\\
	  & \approx-\int_{\partial\Sigma}(\dd^{D-2}x)_{\mu\nu}\xi^{[\mu}\partial_\alpha S^{\nu]\alpha}(g,\delta_{1}g,\delta_{2}g),
	  \label{eq:integrability_general}
\end{align}
equation~\eqref{eq:deltaH_commute} is equivalent to
\begin{align}
    \int_{\partial\Sigma}(\dd^{D-2}x)_{\mu\nu}\xi^{[\mu}\partial_\alpha S^{\nu]\alpha}(g,\delta_{1}g,\delta_{2}g)=0\label{condition1}.
\end{align}
Although this is a necessary condition for the charge to exist, it is also a sufficient condition as long as the space of $g_{\mu\nu}$ has no topological obstruction \cite{Wald_Zoupas_2000}. Therefore, we call Eq.~\eqref{condition1} the integrability condition. 

If the integrability condition is satisfied, the charge at a metric $g_{\mu\nu}$ can be evaluated by the integral along a smooth path from a reference metric $g^{(0)}_{\mu\nu}$ to $g_{\mu\nu}$ in the space of metrics. More precisely, by using an arbitrary one-parameter set of metrics $g_{\mu\nu}(\lambda)$ such that $g_{\mu\nu}(\lambda=0)=g^{(0)}_{\mu\nu}$ and $g_{\mu\nu}(\lambda=1)=g_{\mu\nu}$, the charge at $g_{\mu\nu}$ is evaluated as
\begin{align}
	H[\xi] = \int_{0}^{1}\dd\lambda\int_{\partial \Sigma}(\dd^{D-2}x)_{\mu\nu}(\partial_{\lambda} Q^{\mu\nu}[\xi](g,\partial_{\lambda}g) + 2\xi^{[\mu}\Theta^{\nu]}(g,\partial_{\lambda} g)),\label{eq_charge_int_along_path}
\end{align}
where we have set the reference of the charge $H[\xi]$ so that it vanishes at $g^{(0)}_{\mu\nu}$. 
Since Eq~\eqref{eq:deltaH_commute} is satisfied, the charge in Eq.\eqref{eq_charge_int_along_path} is independent of the choice of path $g_{\mu\nu}(\lambda)$.

\section{Review on a Lie algebra based approach}
\label{sec:review_on_Lie}
In this section, we review the our approach developed in \cite{tomitsuka2021}, where we proposed a guiding principle which helps us to find a non-trivial algebra of the charges. This principle ensures the existence of two elements in the algebra such that their Poisson bracket does not vanish. Therefore, as long as the integrability condition of the charges is satisfied, the transformation generated by the algebra cannot be gauged away. 


In order to investigate the asymptotic symmetries of a background metric $\bar{g}_{\mu\nu}$ of interest with the covariant phase space method, we have to specify (i) the set of metrics which includes $\bar{g}_{\mu\nu}$ and (ii) the set of vector fields which forms a closed algebra. In the following, they are denoted by $\mathcal{S}$ and $\mathcal{A}$, respectively. 
These sets must be chosen such that an element of $\mathcal{S}$ is mapped into itself under any infinitesimal diffeomorphism generated by $\mathcal{A}$. 
Note that only the asymptotic behaviors of the metrics and the vector fields are relevant for the charges. 
In prior studies, such as \cite{brown1986}, it is common to fix the algebra $\mathcal{A}$ as the set of vectors satisfying the asymptotic Killing equation for a given $\mathcal{S}$. 
In this approach, lots of trials and errors are required to find $\mathcal{S}$ such that the integrability condition is satisfied and that the charges form a non-trivial algebra. 

In the Lie algebra based approach proposed in \cite{tomitsuka2021}, an alternative way is adopted to fix $\mathcal{S}$ and $\mathcal{A}$; 
given an algebra $\mathcal{A}$, we define $\mathcal{S}$ by
\begin{align}
    \mathcal{S}\coloneqq \left\{\phi^{*}\bar{g}_{\mu\nu}\middle| \phi\in\left\{ \text{all diffeomorphisms generated by $\mathcal{A}$}\right\}\right\},\label{eq:definition_S}
\end{align}
where $\phi^*$ denotes the pullback. 
In this case, we need to choose $\mathcal{A}$ carefully so that the resulting charges are integrable and form a non-trivial algebra. 
In the rest of this paper, the set $\mathcal{S}$ is always defined by Eq.\eqref{eq:definition_S}. 

There are advantages to adopt the set $\mathcal{S}$ defined in Eq.\eqref{eq:definition_S}. 
First, if $\bar{g}_{\mu\nu}$ is a solution of the Einstein equations, then any element of $\mathcal{S}$ automatically satisfies the Einstein equations. 
In addition, a linearized perturbation $\delta g_{\mu\nu}$ is generated by an infinitesimal diffeomorphism and can be written as
\begin{align}
    \delta g_{\mu\nu}=\pounds_\chi g_{\mu\nu}\label{set2}
\end{align}
with a vector field $\chi\in\mathcal{A}$. In the following, the variation corresponding to such a perturbation is denoted by $\delta_\chi$. 
This property is particularly important to find a candidate of $\mathcal{A}$ with the Lie algebra based method as we will see soon. 
A schematic picture of the set of metrics $\mathcal{S}$ is shown in FIG.~\ref{fig:configuration_space}.
\begin{figure}[htbp]
	\centering
	\includegraphics[width=8.5cm]{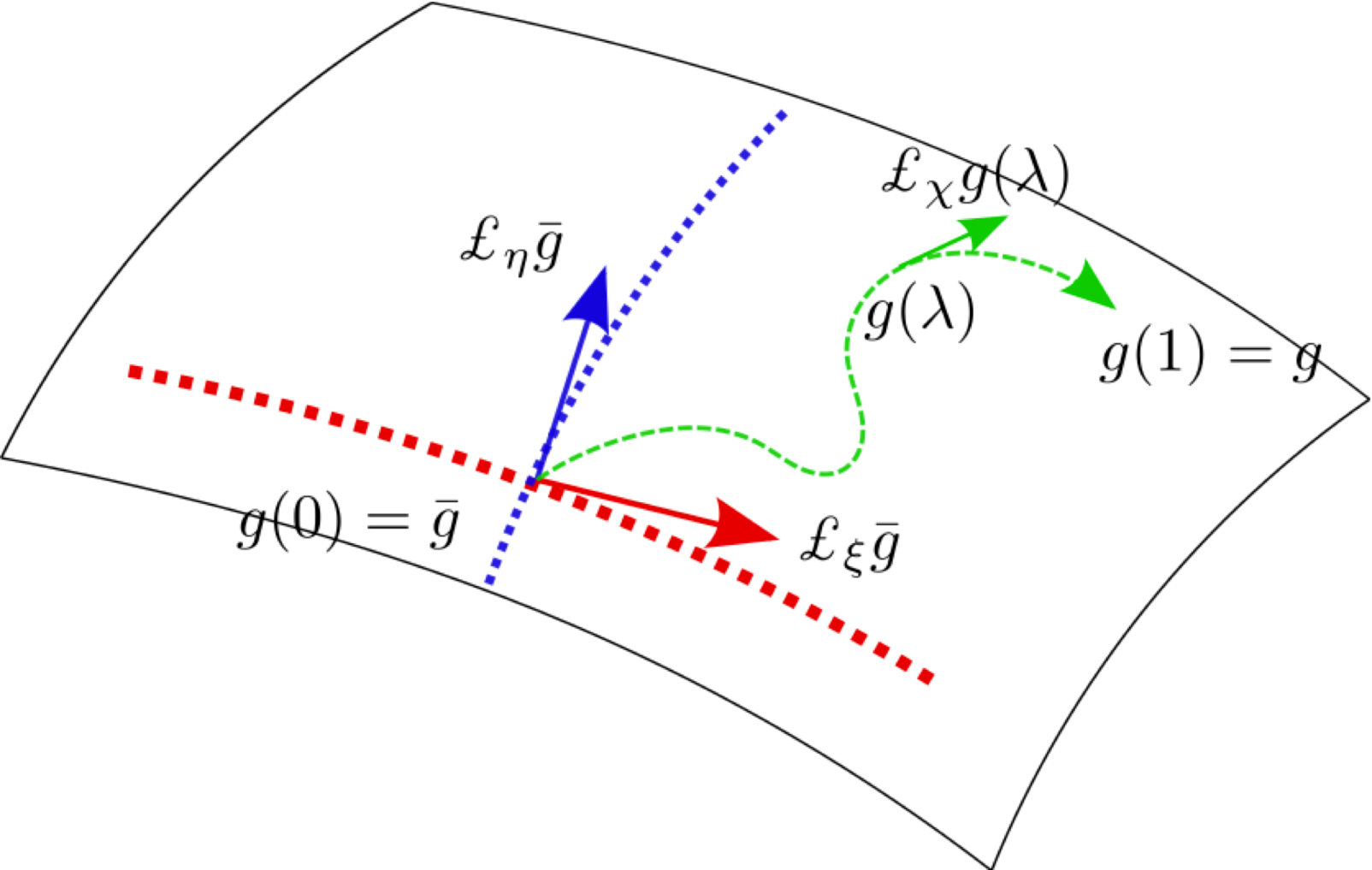}
	\caption{A schematic picture of the set of metrics $\mathcal{S}$ defined in Eq.\eqref{eq:definition_S}.
	Vector fields $\xi$ and $\eta$ are elements of a Lie algebra $\mathcal{A}.$
	All metrics in $\mathcal{S}$ are connected to the background metric $\bar{g}_{\mu\nu}$ by diffeomorphisms generated by $\mathcal{A}$. For any metric $g_{\mu\nu}\in\mathcal{S}$, there exists a smooth path $g_{\mu\nu}(\lambda)$ from $\bar{g}_{\mu\nu}$ to $g_{\mu\nu}$. For any tangent $\delta g_{\mu\nu}(\lambda)$ at a point $g_{\mu\nu}(\lambda)$ in $\mathcal{S}$, there is a vector field $\chi\in\mathcal{A}$ such that $\delta g_{\mu\nu}(\lambda)=\pounds_{\chi}g_{\mu\nu}(\lambda)$.}
	\label{fig:configuration_space}
\end{figure}

Now, let us review the key idea in \cite{tomitsuka2021}, which is helpful to find $\mathcal{A}$ yielding a non-trivial algebra of charges.  
The algebra is non-trivial if
\begin{align}
	\label{delH}
	\exists \xi,\eta\in\mathcal{A}, \exists g_{\mu\nu}\in\mathcal{S},\quad 
	\delta_{\eta}H[\xi]\biggl|_{g_{\mu\nu}}\neq 0,
\end{align}
or equivalently,  $\{H[\xi],H[\eta]\}\Bigl|_{g_{\mu\nu}}\neq 0$.
From Eq.~\eqref{eq_definition_S}, Eq.~\eqref{delH} can be recast into
\begin{align}
    \exists \xi,\eta\in\mathcal{A}, \exists g_{\mu\nu}\in\mathcal{S},
	\int _{\partial\Sigma}(\dd^{D-2}x)_{\mu\nu}S^{\mu\nu}(g,\pounds_\eta g,\pounds_\xi g)  & \neq 0.     \label{non_triviality}
\end{align}
The diffeomorphism associated with the algebra cannot be gauged away if Eq.~\eqref{non_triviality} is satisfied.
Otherwise, all the charges vanish for any metric, implying that the metrics in $\mathcal{S}$ cannot be discriminated by the value of charges and that the diffeomorphisms generated by $\mathcal{A}$ may be gauged away. 

Note that it may be hard to check the condition in Eq.~\eqref{non_triviality} directly since the set of metrics $\mathcal{S}$ depends on $\mathcal{A}$. 
Instead, we adopt a sufficient condition
\begin{align}
    \exists \xi,\eta\in\mathcal{A},
	\int _{\partial \Sigma}(\dd^{D-2}x)_{\mu\nu}S^{\mu\nu}(\bar{g},\pounds_{\eta} \bar{g},\pounds_{\xi} \bar{g}) \neq 0\label{eq_non-triviality_background}
\end{align}
as a guiding principle to fix $\mathcal{A}$. More precisely, we first derive a formula for
\begin{align}
\int _{\partial \Sigma}(\dd^{D-2}x)_{\mu\nu}S^{\mu\nu}(\bar{g},\pounds_{\eta} \bar{g},\pounds_{\xi} \bar{g})    \label{eq:int_S}
\end{align}
for arbitrary vector fields $\xi$ and $\eta$. 
Since Eq.~\eqref{eq_non-triviality_background} can be calculated at $\bar{g}_{\mu\nu}$, we do not need to specify $\mathcal{S}$ nor $\mathcal{A}$ at this point. 
By using it, we then fix two vector fields $\xi$ and $\eta$ so that Eq.\eqref{eq:int_S} does not vanish. 
We define $\mathcal{A}$ as a closed algebra containing $\eta$ and $\xi$, which can be obtained by calculating the commutators of $\xi$ and $\eta$.
The algebra $\mathcal{A}$ defined in this way trivially satisfies Eq.\eqref{non_triviality} and hence the diffeomorphisms generated by $\mathcal{A}$ cannot be gauged away by construction. 

Of course, we also need to impose Eq.~\eqref{condition1} to get integrable charges. This condition can be recast into
\begin{align}
	0=\int_{\partial \Sigma}\left(\dd^{D-2}x\right)_{\mu\nu}\xi^{[\mu}\partial_\alpha S^{\nu]\alpha}\left(g,\pounds_\eta g,\pounds_\chi g\right), \quad \forall \xi,\eta,\chi \in \mathcal{A},\quad \forall g_{\mu\nu}\in\mathcal{S}\label{eq_integrability}
\end{align}
where we have used Eq.~\eqref{set2}. 

For a given background metric $\bar{g}_{\mu\nu}$, 
Eq.~\eqref{eq_non-triviality_background} works as a guiding principle to find non-trivial charges. 
However, there still remains a difficulty in finding integrable charges since we have to choose $\xi$ and $\eta$ so that  Eq.~\eqref{eq_integrability} is also satisfied, which requires trials and errors.
It often takes an effort to check Eq.~\eqref{eq_integrability} for an arbitrary $g_{\mu\nu}\in\mathcal{S}$ since we have to calculate the asymptotic behaviors of $g_{\mu\nu}$ near the boundary. As a necessary condition, in Ref.~\cite{tomitsuka2021}, we adopted Eq.~\eqref{condition1} at the background metric, i.e.,
\begin{align}
    \int_{\partial \Sigma}\left(\dd^{D-2}x\right)_{\mu\nu}\xi^{[\mu}\partial_\alpha S^{\nu]\alpha}\left(\bar{g},\pounds_\eta \bar{g},\pounds_\chi \bar{g}\right)=0, \quad\forall \xi,\eta,\chi \in \mathcal{A}\label{eq_integrability_background}
\end{align}
before checking Eq.~\eqref{eq_integrability} directly. 
This condition can be checked relatively easily since we only need the background metric $\bar{g}_{\mu\nu}$ and the algebra $\mathcal{A}$.
The approach proposed in Ref.~\cite{tomitsuka2021} can be summarized as the following six steps:
\begin{screen}
\begin{description}
    \item[Step~1]
         \quad Fix a background metric $\bar{g}_{\mu\nu}$ of interest.
	\item[Step~2]
	      \quad For the background metric, find two vector fields $\xi$ and $\eta$ satisfying Eq.~\eqref{eq_non-triviality_background}.
	      These are the candidates generating non-trivial diffeomorphisms whose charges are integrable.
	\item[Step~3]
	      \quad Introduce the minimal Lie algebra $\mathcal{A}$ including $\xi$ and $\eta$ by calculating their commutators.
	      Check whether the integrability condition at the background metric, i.e., Eq.~\eqref{eq_integrability_background}, 
	      is satisfied for the algebra $\mathcal{A}$ as a necessary condition for Eq.~\eqref{eq_integrability}.
	      If it holds, go to the next step. Otherwise, go back to Step~2.
	\item[Step~4] 
	\quad Construct the set $\mathcal{S}$ of metrics $g_{\mu\nu}$ which are connected to the background metric $\bar{g}_{\mu\nu}$ via diffeomorphisms generated by $\mathcal{A}$.
	\item[Step~5]
	\quad Check the integrability condition in Eq.~\eqref{condition1}.
	      If it is satisfied, then go to the following step.
	      If not, go back to Step~2.
	\item[Step~6]
	\quad Calculate the charges by using Eq.~\eqref{eq_charge_int_along_path}. Here, we fix the reference metric as the background metric: $g^{(0)}_{\mu\nu}=\bar{g}_{\mu\nu}$.
\end{description}
\end{screen}

In our previous paper \cite{tomitsuka2021}, we only considered the vacuum solutions to the Einstein equation without the cosmological constant. 
We can easily extend the analysis to the solutions with the cosmological constant. 
In this case, Eq.\eqref{non_triviality} is recast into
\begin{align}
&\exists \xi,\eta\in\mathcal{A}, \exists g_{\mu\nu}\in\mathcal{S},\nn \\
&\frac{1}{8\pi G} \int_{\partial\Sigma}\Biggl[(2\nabla^{\alpha}\eta^{\mu}\nabla_{\alpha}\xi^{\nu} -\nabla_{\alpha}\eta^{\alpha}\nabla^{\mu}\xi^{\nu}+\nabla_{\alpha}\xi^{\alpha}\nabla^{\mu}\eta^{\nu}) 
		-C\indices{_{\alpha\beta}^{\mu\nu}}\xi^{\alpha}\eta^{\beta}+ \frac{4\Lambda}{D-1}\xi^{\mu}\eta^{\nu}\Biggl]{\boldsymbol \epsilon}_{\mu\nu} \neq 0. \label{Non}
\end{align}
where $C\indices{_{\alpha\beta}^{\mu\nu}}\coloneqq g^{\mu\gamma}g^{\nu\delta}C_{\alpha\beta\gamma\delta}$ is the Weyl tensor and ${\boldsymbol \epsilon}_{\mu\nu}\coloneqq \sqrt{-g}(\dd ^{D-2}x)_{\mu\nu}$. Note that it can be checked that Eq.~\eqref{Non} is equivalent to Eq.~(39) in Ref.~\cite{tomitsuka2021} if the cosmological constant $\Lambda$ vanishes.

In Step~2 of the above algorithmic protocol, Eq.~\eqref{eq_non-triviality_background} plays a role of a guiding principle to find non-trivial charges. 
In addition, Eq.~\eqref{eq_integrability_background} in Step~3 helps to reduce useless calculations on the charges which turn out not to be integrable. 
An advantage of the above algorithmic protocol is the fact that calculations in Steps~2 and 3 can be done by using only the background metric $\bar{g}_{\mu\nu}$. 
By using this protocol, we have found a new class of symmetries on a Rindler horizon in Ref.~\cite{tomitsuka2021}, which generates position dependent dilatations in  time and in the direction perpendicular to the horizon. We have termed such a transformation superdilatation.
However, there still remain the following hard tasks:
In Step~4, it is required to identify all diffeomorphisms generated by vector fields in $\mathcal{A}$ to obtain $\mathcal{S}$, which is usually difficult. 
Only after this step is completed, the integrability condition can be checked for all metrics in $\mathcal{S}$ in Step~5. 


To overcome this issue, in the next section, we propose a sufficient condition for the charges to be integrable, which can be checked at the background metric $\bar{g}_{\mu\nu}$. 
It enables us to find an algebra $\mathcal{A}$ yielding non-trivial and integrable charges without explicitly calculating diffeomorphisms generated by $\mathcal{A}$ or the metrics in $\mathcal{S}$. 
This is a key advantage of the new approach in this article. 
To calculate the charges explicitly, we still need to identify $\mathcal{A}$ and $\mathcal{S}$. 
However, since the sufficient condition ensures that the charges are integrable, there is no possibility that the efforts in calculating $\mathcal{A}$ and $\mathcal{S}$ are wasted. 

It should be noted that the algebra of charges can be identified without calculating the values of the charges explicitly. 
In fact, the Poisson bracket of the charges satisfies
\begin{align}
    \{H[\xi], H[\eta]\} = H\bigl[[\xi,\eta]\bigl] + K(\xi,\eta),\label{eq:poisson_bracked_general}
\end{align}
where $[\xi,\eta]$ is a commutator of $\xi,\eta$ and $K(\xi,\eta)$ is a constant dependent not on $g
_{\mu\nu}$ but on $\bar{g}_{\mu\nu}$
(see e.g., Ref.~\cite{doi:10.1063/1.2889721}). 
Evaluating the left hand side of Eq.~\eqref{eq:poisson_bracked_general} at the background metric $\bar{g}_{\mu\nu}$, we get $K(\xi,\eta)$ since it is always possible to make the values of charges $H[\chi]\Bigl|_{\bar{g}_{\mu\nu}}$ at the background metric $\bar{g}_{\mu\nu}$ vanish for all $\chi\in\mathcal{A}$. 
If $K(\xi,\eta)$ can be absorbed into charges by shifting them by constants, then the algebra of the charges is isomorphic to $\mathcal{A}$. If not, the algebra of the charges is a central extension of $\mathcal{A}$.
Therefore, we can fully characterize the algebra of charges itself without calculating the diffeomorphisms generated by $\mathcal{A}$ explicitly, overcoming the difficulties in the approach in Ref.~\cite{tomitsuka2021}. 

\section{Integrability condition}
\label{sec:integrability}
In this section, we provide a sufficient condition for the charges to be integrable. This condition can be checked at the background metric, implying that we can obtain integrable charges without calculating the family of metrics $\mathcal{S}$ directly.


Given an algebra $\mathcal{A}$, the integrability condition that the second line in Eq.~\eqref{eq:integrability_general} equals to zero is recast to
\begin{align}
	\int_{\partial\Sigma}(\dd^{D-2}x)_{\mu\nu}\xi^{[\mu}(x)\omega^{\nu]}(g, \pounds_{\eta}g, \pounds_{\chi}g; x)=0\ \ \ \forall \xi, \eta, \chi \in \mathcal{A},\quad \forall g\in\mathcal{S} \label{integ},
\end{align}
where we have used Eq.~\eqref{set2}, $\mathcal{S}$ is the set of metrics defined in Eq.~\eqref{eq:definition_S} and $\omega^{\nu}(g, \delta_{1}g, \delta_{2}g;x)$ is given by
\begin{align}
	\omega^{\nu}(g, \delta_{1}g, \delta_{2}g;x) & =\frac{\sqrt{-g(x)}}{16\pi G}\left(g^{\nu\alpha}(x)\left(g^{\rho\sigma}(x)g^{\beta\gamma}(x)-2g^{\rho\beta}(x)g^{\sigma\gamma}(x)\right)\right. \nn \\
	& \left.	+2g^{\nu\gamma}(x)g^{\rho[\alpha}(x)g^{\sigma]\beta}(x) + g^{\nu\rho}(x)g^{\alpha\beta}(x)g^{\sigma\gamma}(x)
	\right)\delta_{[1}g_{\rho\sigma}(x)\nabla_{\gamma}\delta_{2]}g_{\alpha\beta}(x)
\end{align}
for a solution $g_{\mu\nu}$ of the Einstein equations and linearized perturbations $\delta_1g_{\mu\nu}$ and $\delta_2g_{\mu\nu}$ satisfying the linearized Einstein equations. 
To check whether Eq.\eqref{integ} is satisfied directly, we need the asymptotic behavior of the integrand near the boundary $\partial \Sigma$. By using the well-known duality between a diffeomorphism and a coordinate transformation of tensor fields (see Appendix~\ref{sec:duality} for the details), we derive a formula to calculate the asymptotic behaviors under certain assumptions which will be made below. 

First we introduce our set-up and several assumptions to derive the sufficient condition for the charges to be integrable. 
We fix a $D$-dimensional background spacetime $(M,\bar{g})$ and a Cauchy surface $\Sigma$.
For notational simplicity, we fix a specific coordinate system $\psi : M \to {\mathbb R^{D}}$ in such a way that the Cauchy
surface is characterized by $t = const.$ and that its boundary is specified by $\rho = 0$, where we have defined
\begin{align}
\psi(p) = (y^{0}(p), y^{1}(p), y^{M}(p)) = (t, \rho, \sigma^{M})\ \ \ \ (M=2,\cdots,D-1).
\end{align}
Let $\mathcal{H}$ denote the union of the boundary for all $t$:
\begin{align}
\mathcal{H}:=\{p\in\partial\Sigma_{t} \ \ \text{for some $t$}\},
\end{align}
or equivalently, 
\begin{align}
\mathcal{H}=\{p\in M|y^{1}(p)=0\}.
\end{align}
In this set-up, the integrability condition evaluated at the background metric is given by
\begin{align}
	\int_{\partial\Sigma}(\dd^{D-2} y)_{\mu\nu}\bar{\xi}^{[\mu}(y)\omega^{\nu]}(\bar{g}, \pounds_{\bar{\eta}}\bar{g}, \pounds_{\bar{\chi}}\bar{g}; y) &= \int_{\psi(\partial\Sigma)}\dd \sigma^{2}\dd \sigma^{3}\cdots \dd \sigma^{D} \bar{\xi}^{[t}(y)\omega^{\rho]}(\bar{g}, \pounds_{\bar{\eta}}\bar{g}, \pounds_{\bar{\chi}}\bar{g};y)\nn \\
	&=0 \ \ \ \ \ \forall \bar{\xi}, \bar{\eta}, \bar{\chi} \in \mathcal{A}. \label{backint}
\end{align}
We assume that any diffeomorphism generated by $\mathcal{A}$ does not map a point in the outside (resp. inside) of $\{\Sigma_t\}_t$ to a point in the inside (resp. outside) of $\{\Sigma_t\}_t$. Then, the $\rho$-component of the vector fields generating the diffeomorphisms must vanish on the boundary. 
Thus, we impose the following condition on the asymptotic behaviors of the vector fields:
\begin{align}
	\forall \xi \in {\mathcal A}, \
	\xi^{t}(y) =\mo{1}{},\ \xi^{\rho}(y) =\mo{\rho}{},\ \xi^{M} = \mo{1}{}\ \ \ \ \ \ \ \ (\rho \to 0)
	\label{sufficient1}.
\end{align}

Let us assume that
\begin{align}
	\forall\eta,\chi\in\mathcal{A},\quad \omega^{t}(\bar{g}, \pounds_{\eta}\bar{g}, \pounds_{\chi}\bar{g}; y) &= \mo{1}{},\
	\omega^{\rho}(\bar{g}, \pounds_{\eta}\bar{g}, \pounds_{\chi}\bar{g}; y) = \mo{\rho}{}, \nn \\
	&\omega^{M}(\bar{g}, \pounds_{\eta}\bar{g}, \pounds_{\chi}\bar{g}; y) = \mo{1}{}\qquad (\rho\to 0)
	\label{sufficient2}
\end{align}
hold.  
Under these assumptions, we get
\begin{align}
	\label{sufficient}
    \forall \bar{\xi},\bar{\eta},\bar{\chi}\in\mathcal{A},\quad 
	\bar{\xi}(y)^{[\mu}\omega^{\nu]}(\bar{g}, \pounds_{\bar{\eta}}\bar{g}, \pounds_{\bar{\chi}}\bar{g}; y)
	=
	\begin{pmatrix}
		0           & \mo{\rho}{} & \mo{1}{}  &\cdots&\cdots  & \mo{1}{}    \\
		\mo{\rho}{} & 0           & \mo{\rho}{} &\cdots&\cdots& \mo{\rho}{} \\
		\mo{1}{}    & \mo{\rho}{} & 0           & \mo{1}{}&\cdots &\mo{1}{}  \\
		\vdots & \vdots & \mo{1}{} & \ddots &\ddots & \vdots\\
		\vdots & \vdots & \vdots&\ddots  & \ddots &\mo{1}{}\\
		\mo{1}{}    & \mo{\rho}{} & \mo{1}{}    &\cdots&\mo{1}{}& 0
	\end{pmatrix}.
\end{align}
Since Eq.~\eqref{backint} is clearly satisfied when Eq.\eqref{sufficient} holds, 
Eqs.~\eqref{sufficient1} and \eqref{sufficient2} are a sufficient condition for Eq.~\eqref{backint} to hold.

Next we further show that Eqs.~\eqref{sufficient1} and \eqref{sufficient2} are a sufficient condition for the charges to be integrable at an arbitrary metric, i.e., Eq.~\eqref{integ}. 
Fix a diffeomorphism $\phi : M \to M$ generated by $\mathcal{A}$.
The integrability condition \eqref{integ} at $g = \phi^{*}\bar{g}$ is written as
\begin{align}
	\int_{\partial\Sigma}(\dd^{D-2} x')_{\mu\nu}\xi^{[\mu}(x')\omega^{\nu]}(g, \pounds_{\eta}g, \pounds_{\chi}g; x')=0\ \ \ \forall \xi, \eta, \chi \in \mathcal{A},
	\label{integ2}
\end{align}
where we have adopted another coordinate system $\varphi$, which is related with $\psi$ by
\begin{align}
    \varphi = \psi \circ \phi :p\in\mathcal{M}\mapsto \varphi(p) = (x'^{0}(p),\cdots,x'^{D-1}(p)).
\end{align}
By using Eqs.~\eqref{det}, \eqref{duality}, \eqref{duality1} and \eqref{duality2}, we have
\begin{align}
	\xi^{[\mu}(x'(p))\omega^{\nu]}(g, \pounds_{\eta}g, \pounds_{\chi}g; x'(p)) = \bar{\xi}^{[\mu}(y(\phi(p)))\omega^{\nu]}(\bar{g}, \pounds_{\bar{\eta}}\bar{g}, \pounds_{\bar{\chi}}\bar{g}; y(\phi(p))), \label{duality3}
\end{align}
where the vector field $\bar{\xi}$ is defined by $\bar{\xi} \coloneqq  (\phi^{*})^{-1}\xi $. 
On the other hand, for the algebra $\mathcal{A}$ whose elements satisfy the asymptotic condition in Eq.~\eqref{sufficient1}, we have
\begin{align}
	x'(y) = (\mo{1}{}, \mo{\rho}{}, \mo{1}{},\cdots,\mo{1}{})\quad (\rho\to 0)
	\label{coordinate_asymp}.
\end{align}
See Appendix~\ref{Flow_proof} for proof.
The integral measure in Eq.~\eqref{integ2} is explicitly calculated as
\begin{align}
	(\dd^{D-2} x')_{\mu\nu}\Big|_{\partial \Sigma} & = \frac{1}{(D-2)!2!}\epsilon_{\mu\nu\alpha_{2}\cdots\alpha_{D-1}}e\indices{^{\alpha_{2}}_{M_{2}}}\cdots e\indices{^{\alpha_{D-1}}_{M_{D-1}}} d\sigma^{M_{2}}\wedge \cdots \wedge d\sigma^{M_{D-1}} \label{measure}
\end{align}
where $e\indices{^{\alpha}_{M}}:=\frac{\p{}x'^{\alpha}}{\p{}\sigma^{M}}$.
By using Eq.~\eqref{coordinate_asymp},
the asymptotic behavior of $e\indices{^{\alpha}_{M}}$ is given by
\begin{align}
    \left(e\indices{^{0}_{M}},e\indices{^{1}_{M}},e\indices{^{2}_{M}},\cdots,e\indices{^{D-1}_{M}}\right)=(\mo{1}{}, \mo{\rho}{}, \mo{1}{},\cdots,\mo{1}{})\quad (\rho\to 0)\label{asym_jacobian}
\end{align}
for any $M=2,3,\cdots D-1$. 
By using Eqs.~\eqref{duality3} and \eqref{measure}, the left hand side of Eq.~\eqref{integ2} is proportional to 
\begin{align}
	\int_{\phi(\partial\Sigma)}\bar{\xi}^{[\mu}(y)\omega^{\nu]}(\bar{g}, \pounds_{\bar{\eta}}\bar{g}, \pounds_{\bar{\chi}}\bar{g}; y)\epsilon_{\mu\nu\alpha_{2}\cdots\alpha_{D-1}}e\indices{^{\alpha_{2}}_{M_{2}}}\cdots e\indices{^{\alpha_{D-1}}_{M_{D-1}}} d\sigma^{M_{2}}\wedge \cdots \wedge d\sigma^{M_{D-1}}. \label{intp}
\end{align}
From the asymptotic behaviors of the coordinates in Eq.~\eqref{coordinate_asymp}, any points in $\mathcal{H}$ is mapped into itself by a diffeomorphism $\phi$ generated by $\mathcal{A}$. 
Therefore, the integral region $\phi(\partial\Sigma)$ corresponds to the limit of $\rho\to 0$. 
Note that, since $\epsilon_{\mu\nu\alpha_{2}\cdots\alpha_{D-1}}$ is anti-symmetric under the change in its indices, the integrand in Eq.~\eqref{intp} vanishes except for the contributions coming from the contractions of indices where one of $(\mu,\nu, \alpha_{M_{2}},\cdots, \alpha_{M_{D-1}})$ is $\rho$. Such a contribution is always $\mo{\rho}{}$ since Eqs.~\eqref{sufficient} and \eqref{asym_jacobian} hold.
Thus, we finally get
\begin{align}
	\eqref{intp} \propto \lim_{\rho \to 0}\int_{\phi(\p{}\Sigma)} \mo{\rho}{} \dd \sigma^{2}\cdots \dd \sigma^{D-1}=0
\end{align}
and conclude that Eqs.~\eqref{sufficient1} is also a sufficient condition for the integrability condition to be satisfied at any metric $g_{\mu\nu}$ in $\mathcal{S}$.

Our approach adopted in this paper is summarized in the following 4 steps:
\begin{screen}
\begin{description}
	\item[Step~1]
	\quad Fix a background metric $\bar{g}_{\mu\nu}$ of interest.
	\item[Step~2]
	      \quad For the background metric $\bar{g}_{\mu\nu}$, find two vector fields $\xi$ and $\eta$ with the asymptotic form in Eq.~\eqref{sufficient1}
	      satisfying Eq.~\eqref{eq_non-triviality_background}.
	      These are the candidates of the vector fields which generate non-trivial diffeomorphisms whose charges are integrable.
	\item[Step~3]
	      \quad Introduce the minimal Lie algebra $\mathcal{A}$ including $\xi$ and $\eta$ by calculating their commutators.
	      Check whether Eq.~\eqref{sufficient2} holds.
	      If it does, go to the next step since the charges are integrable. Otherwise, go back to Step~2.
	\item[Step~4]
	    \quad Investigate the algebra of the charges for $\mathcal{A}$ via \eqref{eq_non-triviality_background}.
\end{description}
\end{screen}

A crucial difference between the approach in Ref.~\cite{tomitsuka2021} and the one proposed in this paper is the step where we check the integrability condition. 
In Ref.~\cite{tomitsuka2021}, we checked whether Eq.\eqref{condition1} holds for candidates of vector fields satisfying Eq.\eqref{eq_non-triviality_background}. 
It takes efforts in this step since we need to calculate all the diffeomorphisms generated by the algebra of the vector fields.
Furthermore, these efforts may be wasted since the charges sometimes turn out not to be integrable. 
In contrast, in our new approach, we adopted Eq.\eqref{sufficient2} as a sufficient condition for the charges to be integrable, which can be checked at the bachground metric.
It is much easier to check Eq.\eqref{sufficient2} than Eq.\eqref{condition1} since we do not need to identify the diffeomorphisms generated by the algebra of the vector fields.

As a demonstration, we investigate asymptotic symmetries on Killing horizons in the following section. 
Adopting our approach, we find that a class of supertranslation, superrotation and superdilatation yields a non-trivial and integrable algebra of charges with a central extension. 


\section{Asymptotic symmetries on Killing horizon}
\label{sec:Killing}
Let us investigate the asymptotic symmetries at a Killing horizon of a spacetime with our new approach developed in the last section. We will find a new class of asymptotic symmetries and show that the algebra of the corresponding charges is a central extension of the algebra  of vector fields generating the transformations of the symmetries.

\noindent
\underline{Step~1}

Here, we adopt the following $D$-dimensional metric as the background metric:
\begin{align}
	(\bar{g}_{\mu\nu}) =
	\begin{pmatrix}
		-\kappa^{2} \rho^{2}+ \mo{\rho}{4} & \mo{\rho}{4}     & f_{t\psi}\rho^{2} + \mo{\rho}{4} & f_{tA}\rho^{2} + \mo{\rho}{4} \\
		\mo{\rho}{4}                                            & 1 + \mo{\rho}{2} & \mo{\rho}{4}                     & \mo{\rho}{3}                  \\
		f_{t\psi}\rho^{2} + \mo{\rho}{4}                        & \mo{\rho}{4}     & f_{\psi\psi} + \mo{\rho}{2}      & \mo{\rho}{2}                  \\
		f_{tA}\rho^{2} + \mo{\rho}{4}                           & \mo{\rho}{3}     & \mo{\rho}{2}                     & \Omega_{AB} + \mo{\rho}{2}
	\end{pmatrix}\ \ \ (\rho \to 0)
	\label{backgroundmetric}
\end{align}
in the coordinate $(t,\rho, \psi,\theta^{A})$ for $A=3,\cdots,D-1$, where all coefficient functions $f_{t\psi},f_{tA}, f_{\psi\psi}$ and $\Omega_{AB}$ depend on $\theta^{A}$ while $\kappa$ is a constant. 
We assume that the coefficient functions and $\kappa$ are fixed so that the metric satisfies the Einstein equations.
This class of metrics contains important spacetimes, for example, de-Sitter spacetime and the Kerr spacetime. 
It is known that the asymptotic behavior of the metric near the Killing horizon located at $\rho=0$ is given by Eq.\eqref{backgroundmetric} and that the Cauchy surface is characterized by $t = const.$ \cite{PhysRevLett.125.241302}.

\noindent
\underline{Step~2}

Next we consider two vector fields $\xi$ and $\eta$ which have the asymptotic forms given by Eq.~\eqref{sufficient1}: 
\begin{align}
	\xi^{\mu}  & = (X^{t}(t,\psi,\theta^{A})+\mo{\rho}{},X^{\rho}(t,\psi,\theta^{A})\rho+\mo{\rho}{2}, X^{\psi}(t,\psi,\theta^{A})+\mo{\rho}{}, X^{A}(t,\psi,\theta^{A})+\mo{\rho}{}), \label{vectors1}\\
	\eta^{\mu} & = (Y^{t}(t,\psi,\theta^{A})+\mo{\rho}{},Y^{\rho}(t,\psi,\theta^{A})\rho+\mo{\rho}{2}, Y^{\psi}(t,\psi,\theta^{A})+\mo{\rho}{}, Y^{A}(t,\psi,\theta^{A})+\mo{\rho}{}) 
	\label{vectors2}
\end{align}
as $\rho \to 0$, where all coefficients are arbitrary functions of $t,\psi$ and $\theta^{A}$.
For the metric \eqref{backgroundmetric}, vector fields \eqref{vectors1} and \eqref{vectors2}, our guiding principle in Eq.~\eqref{Non} can be calculated as follows:
\begin{align}
	\frac{1}{8\pi G}\int_{\p{}{\Sigma}}\frac{2\sqrt{\Omega f_{\psi\psi}}}{\kappa}
	\Bigg[
	\frac12 \p{t}Y^{\rho}\p{t}X^{t} + &D_{M}Y^{M}\left(\kappa^{2}X^{t} - f_{tN}X^{N} + \frac12 \p{t}X^{\rho}\right)+ \p{A}f_{t\psi}X^{\psi}Y^{A} \nn \\
	&+ \left(\p{B}f_{tA}-\p{A}f_{tB}\right)X^{A}Y^{B}
	- (X \leftrightarrow Y)
	\Bigg] \dd \sigma^{2}\cdots \dd \sigma^{D-1} \neq 0 \label{non_trivial_Killing}
\end{align}
where $M,N=2,\cdots, D-1$ and $D_{M}$ denotes the covariant derivative on the $(D-2)$-dimensional hypersurface characterized by $t=const.$ and $\rho=const.$.

In Ref.~\cite{tomitsuka2021}, we investigated the following set of vector fields
\begin{align}
\xi^{t} &= tT_{1}(x^{M})+\mo{\rho}{2},\ 
\xi^{\rho}= \mo{\rho}{2},\ 
\xi^{M}= \mo{\rho}{2},\\
\eta^{t} &=\mo{\rho}{2},\ 
\eta^{\rho}=tT_{2}(x^M)\rho+\mo{\rho}{2},\ 
\eta^{M}= \mo{\rho}{2},
\end{align}
where $T_{1},T_{2}$ are arbitrary functions of $x^{M}$, which generate superdilatations. This is one of sets satisfying Eq.~\eqref{non_trivial_Killing}.
On the other hand, for given functions $T(x^{M})$ and $V^{M}(x^{N})$ of $x^{M}$, the vector fields defined by
\begin{align}
	\xi^{t}  & = T(x^{M})+\mo{\rho}{2},\ \xi^{\rho} = \mo{\rho}{2},\ 
	\xi^{M}= \mo{\rho}{2},\label{eq:xi_st_sr}\\
	\eta^{t} & = \mo{\rho}{2},\ \eta^{\rho} = \mo{\rho}{2},\ \eta^{M} = V^{M}(x^{N})+\mo{\rho}{2}\label{eq:eta_st_sr}
\end{align}
also satisfy Eq.\eqref{non_trivial_Killing}. 
In fact, this set of vector fields generates a well-known class of transformations called supertranslations and superrotations. 
See Appendix~\ref{sec:st_and_sr} for a comment on the integrability of the charges for this algebra. 
As a first trial, let us analyze a simple algebra containing the above two known cases, which is given by
\begin{align}
\xi^{t} &= F_{1}(x^{M}) + tG_{1}(x^{M})+\mo{\rho}{2},\ 
\xi^{\rho}= (H_{1}(x^{M}) + tJ_{1}(x^M))\rho+\mo{\rho}{2},\ 
\xi^{M}= K_{1}^{M}(x^{N}) + \mo{\rho}{2},\label{eq:asymptotic_vec_sd_st_sr1} \\
\eta^{t} &= F_{2}(x^{M}) + tG_{2}(x^{M})+\mo{\rho}{2},\ 
\eta^{\rho}= (H_{2}(x^{M}) + tJ_{2}(x^M))\rho+\mo{\rho}{2},\ 
\eta^{M}= K_{2}^{M}(x^{N}) + \mo{\rho}{2}\label{eq:asymptotic_vec_sd_st_sr2}
\end{align}
in the rest of this section, where $F_{i}(x^{M}), G_{i}(x^{M}), H_{i}(x^{M}), J_{i}(x^{M})$ and $K^{M}_{i}(x^{N})$ are arbitrary functions of $x^{M}$.

\noindent
\underline{Step~3}

For an arbitrary set of vector fields with asymptotic behavior in Eq.~\eqref{vectors2}, the pre-symplectic current at the background metric given in Eq.~\eqref{onshell_pre} can be calculated as 
\begin{subequations}
\begin{align}
\label{omega_Killing_t}
\omega^{t}(\bar{g}, \pounds_{\eta}\bar{g}, \pounds_{\xi}\bar{g}) &\approx  \p{M}\left(-\frac{\sqrt{\Omega f_{\psi\psi}}}{2\kappa\rho}\left[\p{t}X^{M}\left(\p{t}Y^{t}-D_{N}Y^{N}\right) - (X \leftrightarrow Y) \right] \right) + \mo{1}{}. \\
\label{omega_Killing_rho}
\omega^{\rho}(\bar{g}, \pounds_{\eta}\bar{g}, \pounds_{\xi}\bar{g}) &\approx -\frac{\sqrt{\Omega f_{\psi\psi}}}{\kappa}
	\p{t}\Bigg(
	\frac12 \p{t}Y^{\rho}\p{t}X^{t} + D_{M}Y^{M}(\kappa^{2}X^{t} \nn \\
	&- f_{tM}X^{M} + \frac12 \p{t}X^{\rho})	
+ \p{A}f_{t\psi}X^{\psi}Y^{A} + \left(\p{B}f_{tA}-\p{A}f_{tB}\right)X^{A}Y^{B}
	- (X \leftrightarrow Y)
	\Bigg) \nn \\
&	+ \p{M}\left( \frac{\sqrt{\Omega f_{\psi\psi}}}{\kappa}\left[\left(-\kappa^{2}Y^{t} + f_{tN}Y^{N} - \p{t}Y^{\rho}\right)\p{t}X^{M} - (X \leftrightarrow Y)\right] \right)+ \mo{\rho}{} \\
\label{omega_Killing_M}
\omega^{M}(\bar{g}, \pounds_{\eta}\bar{g}, \pounds_{\xi}\bar{g}) &\approx \frac{\sqrt{\Omega f_{\psi\psi}}}{2\kappa\rho}\p{t}\Big(\p{t}X^{M}\left(\p{t}Y^{t}-D_{N}Y^{N}\right) - (X \leftrightarrow Y) \Big) \nn\\
&+\p{N}\left(-\frac{\sqrt{\Omega f_{\psi\psi}}}{\kappa \rho}\left[
\p{t}Y^{M}\p{t}X^{N} - (X \leftrightarrow Y)\right] \right)
+ \mo{1}{}.
\end{align}
\end{subequations}
for $\rho \to 0$.

The components of the commutators of the vector fields in Eqs.~\eqref{eq:asymptotic_vec_sd_st_sr1} and \eqref{eq:asymptotic_vec_sd_st_sr2} are calculated as
\begin{align}
[\xi, \eta]^{t} &= (F_{1}G_{2} - G_{1}F_{2} + K_{1}^{M}\p{M}F_{2} - K_{2}^{M}\p{M}F_{1}
) + t(K_{1}^{M}\p{M}G_{2} - K_{2}^{M}\p{M}G_{1})+\mo{\rho}{2} \nn \\
[\xi,\eta]^{\rho}& = \left\{(F_{1}J_{2} -J_{1}F_{2} + K_{1}^{M}\p{M}H_{2} - K_{2}^{M}\p{M}H_{1}) + t(G_{1}J_{2} - J_{1}G_{2} +K_{1}^{M}\p{M}J_{2} - K_{2}^{M}\p{M}J_{1})\right\}\rho+\mo{\rho}{2} \nn \\
[\xi,\eta]^{M}&= \left(K_{1}^{N}\p{N}K_{2}^{M} - K_{2}^{N}\p{N}K_{1}^{M}\right) + \mo{\rho}{2} 
\end{align}
for $\rho\to 0$. 
Thus, let us define the closed algebra $\mathcal{A'}$ including $\xi,\eta$
\begin{align}
\mathcal{A'}&  \nn\\
&\coloneqq \left\{
V=\left(
F(x^{M}) + tG(x^{M})+\mo{\rho}{2},
\rho\left(H(x^{M}) + tJ(x^{M})\right)+\mo{\rho}{2},
K^{M}(x^{N})+ \mo{\rho}{2}
\right)\right\}.
\end{align}
In this case, since we have 
\begin{align}
	\omega^{t}(\bar{g}, \pounds_{\eta}\bar{g}, \pounds_{\xi}\bar{g}) = \mo{1}{},\ \omega^{\rho}(\bar{g}, \pounds_{\eta}\bar{g}, 		\pounds_{\xi}\bar{g}) = \mo{1}{},\ \omega^{M}(\bar{g}, \pounds_{\eta}\bar{g}, \pounds_{\xi}\bar{g}) = \mo{1}{}\quad (\rho \to 0)\ \ \ \ \forall 		\eta, \xi \in {\mathcal A'}
\end{align}
from Eqs.~\eqref{omega_Killing_t}$\sim$\eqref{omega_Killing_M}, 
Eq.~\eqref{sufficient2} is not satisfied. 
Thus, $\mathcal{A}'$ is not suitable for our purpose. 

From Eq.~\eqref{omega_Killing_rho}, it immediately turns out that if we impose an additional condition
\begin{align}
    D_{M}K^{M} = 0
    \label{divergencelss},
\end{align}
then we get $\omega^{\rho}(\bar{g}, \pounds_{\eta}\bar{g}, 		\pounds_{\xi}\bar{g}) = \mo{\rho}{}$ and hence Eq.~\eqref{sufficient2} is satisfied.
This condition in Eq.~\eqref{divergencelss} means that we pick up only a divergenceless part in superrotation.
Since
\begin{align*}
D_{M}(K_{1}^{N}\p{N}K_{2}^{M}- K_{2}^{N}\p{N}K_{1}^{M}) &= 
D_{M}K_{1}^{N}D_{N}K_{2}^{M} - D_{M}K_{2}^{N}D_{N}K_{1}^{M} +
K_{1}^{N}D_{M}D_{N}K_{2}^{M}  - K_{2}^{N}D_{M}D_{N}K_{1}^{M}  \\
&= K_{1}^{N}R_{LN}K_{2}^{L} + K_{1}^{N}D_{N}D_{M}K_{2}^{M} - K_{2}^{N}R_{LN}K_{1}^{L} - K_{2}^{N}D_{N}D_{M}K_{1}^{M} \\
&=0,
\end{align*}
holds,
the algebra
\begin{align}
&\mathcal{A}  \nn\\
&\coloneqq \left\{
V=\left(
F(x^{M}) + tG(x^{M})+\mo{\rho}{2},
\rho\left(H(x^{M}) + tJ(x^{M})\right)+\mo{\rho}{2},
K^{M}(x^{N})+ \mo{\rho}{2}
\right)
\mid D_{M}K^{M}=0 \right\}
\label{A'}
\end{align}
is closed. Therefore, instead of $\mathcal{A}'$, we hereafter adopt $\mathcal{A}$. 
Since Eqs.~\eqref{eq_non-triviality_background} and \eqref{sufficient2} are satisfied for $\mathcal{A}$, the charges are integrable and form a non-trivial algebra. 

\noindent
\underline{Step~4}

Let us investigate the algebra of charges for $\mathcal{A}$.
For simplicity, in the following, we will analyze
\begin{align}
	(\bar{g}_{\mu\nu}) =
	\begin{pmatrix}
		-\kappa^{2} \rho^{2}+ \mo{\rho}{4} & \mo{\rho}{4}     & f_{t\theta}\rho^{2} + \mo{\rho}{4} & f_{t\phi}\rho^{2} + \mo{\rho}{4} \\
		\mo{\rho}{4}                                            & 1 + \mo{\rho}{2} & \mo{\rho}{4}                     & \mo{\rho}{3}                  \\
		f_{t\theta}\rho^{2} + \mo{\rho}{4}                        & \mo{\rho}{4}     & A + \mo{\rho}{2}      & \mo{\rho}{2}                  \\
		f_{t\phi}\rho^{2} + \mo{\rho}{4}                           & \mo{\rho}{3}     & \mo{\rho}{2}                     & A\sin^{2}\theta + \mo{\rho}{2}
	\end{pmatrix}\ \ \ 
	\label{g_killing}
\end{align}
as $\rho\to 0$
in the coordinate system $(t,\rho,\theta,\phi)\ \ (0\leq \theta \leq \pi, 0 \leq \phi \leq 2\pi)$ for $D=4$.
In this case, the induced metric on the horizon is given by $ds^{2}|_{\p{}\Sigma} = A(d\theta^{2} + \sin^{2}\theta d\phi^{2})$, where $A>0$ is a parameter describing the area of the horizon.

Functions characterizing an element in $\mathcal{A}$ in Eq.~\eqref{A'} can be expanded as follows:
\begin{gather}
\begin{align}
&F(\theta, \phi) = \sum_{lm}a_{lm}Y_{lm}(\theta,\phi),~~G(\theta, \phi) = \sum_{lm}b_{lm}Y_{lm}(\theta,\phi), \\
&H(\theta, \phi) = \sum_{lm}c_{lm}Y_{lm}(\theta,\phi),~~J(\theta, \phi) = \sum_{lm}d_{lm}Y_{lm}(\theta,\phi), 
\end{align}\\
K^{A}(\theta,\phi) = -\frac{1}{\sin\theta}\epsilon^{AB}\p{B}\Psi(\theta,\phi), \  \Psi(\theta, \phi)=\sum_{lm}e_{lm}Y_{lm}(\theta,\phi),
\end{gather}
where 
\begin{align}
    Y_{lm}(\theta, \phi) = (-1)^{m}\sqrt{\frac{(2l+1)}{4\pi}\frac{(l-m)!}{(l+m)!}}P_{l}^{m}(\cos\theta)e^{im\phi}
\end{align}
is the spherical harmonics, $P^{m}_{l}(\cos\theta)$ is the associated Legendre polynomials and 
\begin{align}
    \epsilon^{\theta\phi} &= -\epsilon^{\phi\theta} = 1,\\
    \epsilon^{\theta\theta} &= \epsilon^{\phi\phi} =0.
\end{align}

All the independent generators are listed as
\begin{subequations}
\begin{align}
&J_{lm}^{(t,0)} = Y_{lm}\p{t} , \\
&J_{lm}^{(t,1)} = tY_{lm}\p{t}, \\
&J_{lm}^{(\rho, 0)} = \rho Y_{lm}\p{\rho}, \\
&J_{lm}^{(\rho,1)} = t\rho Y_{lm}\p{\rho}, \\
&J_{lm}^{(R)} =\frac{1}{\sin\theta}\left(\p{\theta}Y_{lm}\p{\phi} - \p{\phi}Y_{lm}\p{\theta} \right) ,
\end{align}
\end{subequations}
where we have omitted $\mathcal{O}(\rho^2)$ in each component of the generators since it does not affect the algebraic structure nor the calculation on the constant term $K(\xi,\eta)$ in Eq.~\eqref{eq:poisson_bracked_general}. 
Their commutators are calculated as 
\begin{subequations}
\begin{align}
\label{vec_algebra}
&[J^{(t,0)}_{lm}, J^{(t,0)}_{l'm'}]  = 0,
~~~[J^{(t,0)}_{lm}, J^{(t,1)}_{l'm'}] = \sum G^{l''m''}_{lml'm'}J^{(t,0)}_{l''m''}, \\
&[J^{(t,0)}_{lm}, J^{(\rho,0)}_{l'm'}]  =0,
~~~[J^{(t,0)}_{lm}, J^{(\rho,1)}_{l'm'}] = \sum G^{l''m''}_{lml'm'}J^{(\rho,0)}_{l''m''}, \\
&[J^{(t,0)}_{lm}, J^{(R)}_{l'm'}] = -\sum C_{lml'm'}^{l''m''}J^{(t,0)}_{l''m''},
\\
&[J^{(t,1)}_{lm}, J^{(t,1)}_{l'm'}] =0,
~~~[J^{(t,1)}_{lm}, J^{(\rho,0)}_{l'm'}] = 0,\\
&[J^{(t,1)}_{lm}, J^{(\rho,1)}_{l'm'}] = \sum G^{l''m''}_{lml'm'}J^{(\rho,1)}_{l''m''},\label{vec_algebra_nontrivial} \\
&[J^{(t,1)}_{lm}, J^{(R)}_{l'm'}] = -\sum C^{l''m''}_{lml'm'}J^{(t,1)}_{l''m''}, \\
&[J^{(\rho,0)}_{lm}, J^{(\rho,0)}_{l'm'}] =0,
~~~[J^{(\rho,0)}_{lm}, J^{(\rho,1)}_{l'm'}] = 0,\\
&[J^{(\rho,0)}_{lm}, J^{(R)}_{l'm'}] = -\sum C^{l''m''}_{lml'm'}J^{(\rho,0)}_{l''m''}, \\
&[J^{(\rho,1)}_{lm}, J^{(\rho,1)}_{l'm'}] = 0,\\
&[J^{(\rho,1)}_{lm}, J^{(R)}_{l'm'}] = -\sum C^{l''m''}_{lml'm'}J^{(\rho,1)}_{l''m''}, \\
&[J^{(R)}_{lm}, J^{(R)}_{l'm'}] = \sum C^{l''m''}_{lml'm'}J^{(R)}_{l''m''},
\label{vec_algebra_last}
\end{align}
\end{subequations} 
where the structure constants $G^{l''m''}_{lml'm'}$ and $C^{l''m''}_{lml'm'}$ satisfy the following relations
\begin{align}
Y_{lm}Y_{l'm'} = \sum_{l''m''}G^{l''m''}_{lml'm'}Y_{l''m''},&~~~G^{l''m''}_{lml'm'} = G^{l''m''}_{l'm'lm}, \\
\frac{1}{\sin\theta}\left(\p{\theta}Y_{lm}\p{\phi}Y_{l'm'} - \p{\phi}Y_{lm}\p{\theta}Y_{l'm'} \right) = \sum_{l''m''} C_{lml'm'}^{l''m''}Y_{l''m''},&~~~C_{lml'm'}^{l''m''} = - C_{l'm'lm}^{l''m''}.
\end{align}
From Eq.\eqref{non_trivial_Killing}, we find that there are two non-vanishing Poisson brackets evaluated at the background metric.
One of them is
\begin{align}
\left\{H[J_{lm}^{(t,1)}], H[J_{l'm'}^{(\rho,1)}]\right\}\Big|_{\bar{g}} =\frac{A}{8\pi G\kappa} \sum_{l''m''} G_{lml'm'}^{l''m''}\int_{0}^{2\pi}\int_{0}^{\pi} Y_{l''m''}\sin\theta d\theta d\phi,\label{eq:Po_t1_rho1}
\end{align}
while the other is
\begin{align}
\left\{H[J_{lm}^{(R)}], H[J_{l'm'}^{(R)}]\right\}\Big|_{\bar{g}} =\frac{A}{8\pi G\kappa} \sum_{l''m''} C_{lml'm'}^{l''m''}\int_{0}^{2\pi}\int_{0}^{\pi} 2\p{\phi}f_{t\theta}Y_{l''m''}d\theta d\phi.\label{eq:Po_RR}
\end{align}

By using these formulas, let us investigate whether the algebra of the charges is a central extension of the algebra of the vector fields. 
For the latter Poisson bracket in Eq.\eqref{eq:Po_RR}, shifting the charge by a constant as
\begin{align}
H'[J_{lm}^{(R)}] \coloneqq
H[J_{lm}^{(R)}]
+\frac{A}{8\pi G\kappa}\int_{0}^{2\pi}\int_{0}^{\pi} 2\p{\phi}f_{t\theta}Y_{lm}d\theta d\phi,
\end{align}
Eq.\eqref{eq:Po_RR} can be rewritten as
\begin{align}
\left\{H'[J_{lm}^{(R)}], H'[J_{l'm'}^{(R)}]\right\} = \sum_{l''m''}C^{l''m''}_{lml'm'}
H'[J_{l''m''}^{(R)}].
\end{align}
This redefinition of the charge does not affect other Poisson brackets. 
On the other hand, for the former one in Eq.~\eqref{eq:Po_t1_rho1}, we may redefine
\begin{align}
H'[J_{lm}^{(\rho,1)}] \coloneqq
H[J_{lm}^{(\rho,1)}]
+\frac{A}{8\pi G\kappa}\int_{0}^{2\pi}\int_{0}^{\pi} Y_{lm}\sin\theta d\theta d\phi,
\end{align} so that Eq.~\eqref{eq:Po_t1_rho1} is recast into
\begin{align}
\left\{H[J_{lm}^{(t,1)}], H'[J_{l'm'}^{(\rho,1)}]\right\} = \sum_{l''m''}G^{l''m''}_{lml'm'}
H'[J_{l''m''}^{(\rho,1)}].
\end{align}
However, since $\left\{H[J_{lm}^{(\rho,1)}], H[J_{l'm'}^{(R)}]\right\}\Big|_{\bar{g}} =0$, this redefinition affects another Poisson bracket in such a way that
\begin{align}
\left\{H'[J_{lm}^{(\rho,1)}], H[J_{l'm'}^{(R)}]\right\} &= -\sum_{l''m''}C^{l''m''}_{lml'm'}H[J_{l''m''}^{(\rho,1)}]  \nn \\
&= -\sum_{l''m''}C^{l''m''}_{lml'm'}\left(H'[J_{l''m''}^{(\rho,1)}] -\frac{A}{8\pi G\kappa}\int_{0}^{2\pi}\int_{0}^{\pi} Y_{l''m''}\sin\theta d\theta d\phi\right)
\end{align}
holds. 
Thus, these constants cannot be absorbed to the generators by redefinition. 
They are calculated as
\begin{align}
K_{lml'm'} \coloneqq  \left\{H[J_{lm}^{(t,1)}], H[J_{l'm'}^{(\rho,1)}]\right\}\Big|_{\bar{g}} &= \frac{A}{8\pi G\kappa}\int_{0}^{2\pi}\int_{0}^{\pi}\sin\theta
Y_{lm}Y_{l'm'}d\theta d\phi \nn \\
&=\frac{A}{8\pi G\kappa}(-1)^{m}\int_{0}^{2\pi}\int_{0}^{\pi}\sin\theta
Y_{lm}Y^{*}_{l'(-m')}d\theta d\phi \nn \\
&=\frac{A}{8\pi G\kappa}(-1)^{m}\delta_{ll'}\delta_{m(-m')}.
\end{align}

Summarizing the above arguments, we finally get the following charge algebra:
\begin{align}
&\left\{H[J_{lm}^{(t,1)}], H[J_{l'm'}^{(\rho,1)}]\right\} = \sum_{l''m''}G^{l''m''}_{lml'm'}H[J_{l''m''}^{(\rho,1)}]
+\frac{A}{8\pi G\kappa}(-1)^{m}\delta_{ll'}\delta_{m(-m')},
\label{result_1}
\\
&\text{others are isomorphic to }\mathcal{A}\text{ in Eqs.~\eqref{vec_algebra}-\eqref{vec_algebra_last} except for $\eqref{vec_algebra_nontrivial}$}.
\label{result_2}
\end{align}
Since $A\neq 0$, the algebra of the charges is a central extension of $\mathcal{A}$. 
Equations~\eqref{result_1} and \eqref{result_2} are the main results in this section.

\section{Summary}
\label{sec:summary}
In this paper, we developed a new approach to investigate asymptotic symmetries by modifying the protocol proposed in Ref.~\cite{tomitsuka2021} by the authors of this paper and a collaborator. 
The key ingredient of our approach is making use of Eqs.~\eqref{eq_non-triviality_background} and \eqref{sufficient2} to find the algebra $\mathcal{A}$ of vector fields that generates transformations of asymptotic symmetries with non-trivial and integrable charges. 
As we have seen in Sec.~\ref{sec:integrability}, 
Eq.~\eqref{sufficient2} provides a sufficient condition for the charges to be integrable, which can be checked at the background metric. 
This is a significant difference between the modified approach and the original one in Ref.~\cite{tomitsuka2021}, which saves the efforts of calculating all the diffeomorphisms generated by $\mathcal{A}$ required in the latter approach. 
As is mentioned in Sec.~\ref{sec:review_on_Lie}, the Poisson brackets of the charges can be calculated at the background metric and hence the algebra of the charges can be fully identified without calculating the diffeomorphisms generated by $\mathcal{A}$ explicitly. 

In Sec.~\ref{sec:Killing}, as a demonstration of our approach, we have investigated asymptotic symmetries of spacetimes with the Killing horizon with metrics in Eq.~\eqref{backgroundmetric}.  
We found that a new algebra of supertranslations, superrotations and superdilatations in Eq.~\eqref{A'} yields a non-trivial algebra of integrable charges. 
It is proven that for the algebra in Eq.~\eqref{A'}, we have to eliminate rotationless part of superrotations to obtain integrable charges. 
As a particular example, for $(1+3)$-dimensional spacetime with metrics in Eq.\eqref{g_killing}, we explicitly calculated the algebra of charges, which is shown to be a central extension of the algebra of the vector fields. 

It should be emphasized that our approach can be applied to any spacetime as long as we consider the diffeomorphisms which do not shift the boundary on which charges are defined.
In particular, as we have mentioned in Introduction, microstates classified by asymptotic symmetries on a horizon are a possible origin of the Bekenstein-Hawking entropy. 
Our algorithmic approach is powerful to list such asymptotic symmetries. 
A discovery of new asymptotic symmetries will lead a better understanding of the nature of gravity and the spacetime structures, 
as the asymptotic symmetries in anti-de Sitter spacetime found in Ref.~\cite{brown1986} led to the development of the AdS/CFT correspondence \cite{AdS/CFT}.

Of course, it should be noted that there may be asymptotic symmetries which cannot be found in our approach since Eqs.~\eqref{eq_non-triviality_background} and \eqref{sufficient2} are sufficient conditions for the charges to be integrable and form a non-trivial algebra. 
Nevertheless, we expect that our approach proposed in this paper is helpful to find new asymptotic symmetries as we have demonstrated the example in Sec.~\ref{sec:Killing}. 

\section*{Acknowledgment}
    The authors thank Masahiro Hotta for useful discussions.
    This research was partially supported by Graduate Program on Physics for the Universe of Tohoku University (T.T.) and by JST SPRING, Grant Number JPMJSP2114 (T.T).


\let\doi\relax

\appendix

\section{Duality between a diffeomorphism and a coordinate transformation of tensor fields}\label{sec:duality}

We here make a brief review of the duality between a diffeomorphism and a coordinate transformation of tensor fields.
Let $M$ and $N$ be $D$-dimensional manifolds.
We consider a $C^\infty$ map $\phi:M\to N$ and the pull-back $g=\phi^{*}\bar{g}$.
We take charts $(U, \varphi)$ around $p \in U \subset M$ and  $(V, \psi)$ around $q=\phi(p)\in V \subset N$.
Each coordinate system is denoted by
\begin{align}
	\varphi(p) = (x^{0}(p),\cdots,x^{D-1}(p)) \\
	\psi(q) = (y^{0}(q),\cdots, y^{D-1}(q)).
\end{align}
The components of the metrics $g$ and $\bar{g}$ are related as
\begin{align}
	g_{\mu\nu}(x(p)) = \bar{g}_{\rho\sigma}(y(q))\frac{\p{}y^{\rho}}{\p{}x^{\mu}}\frac{\p{}y^{\sigma}}{\p{}x^{\nu}},\label{eq:components_metric_pullback}
\end{align}
where $g|_{p} = g_{\mu\nu}(x(p))dx^{\mu}|_{p} \otimes dx^{\nu}|_{p}$ and $\left.\bar{g}\right|_{q} = \left.\bar{g}_{\rho\sigma}(y(q))dy^{\rho}\right|_{q}\otimes \left.dy^{\sigma}\right|_{q}$.
Since $\psi \circ\phi$ is a smooth function $M \to {\mathbb R}^{D}$, we can introduce a new coordinate system around $p\in M$
\begin{align}
	\psi\circ\phi(p) = (x'^{0}(p),\cdots,x'^{D-1}(p)).
\end{align}
From Eq.\eqref{eq:components_metric_pullback}, the metric $g$ satisfies
\begin{align}
	\bar{g}_{\mu\nu}(y(q))=g_{\mu\nu}(x'(p)),
\end{align}
where $\left.g\right|_{p} = g_{\mu\nu}(x'(p))\left.dx'^{\mu}\right|_{p}\otimes \left. dx'^{\nu}\right|_{p} $.
This means that
the components of $\bar{g}|_{q} \in T^{*}_{q}N\otimes T^{*}_{q}N$ in a coordinate system $\psi:N \to {\mathbb R^{D}}$ are equal to the components of $g|_{p} \in T^{*}_{p}M \otimes T^{*}_{p}M$ in another coordinate system $\psi \circ \phi : M \to {\mathbb R^{D}}$.
Note that
\begin{align}
	\sqrt{-\bar{g}(y(q))} = \sqrt{-g(x'(p))} \label{det}
\end{align}
also holds, where $g(x'(p))$ and $\bar{g}(y(q))$ are the determinants of the metrics. 

In general, for the pull-backed $(r,s)$-tensor $T = \phi^{*}\bar{T}$, we have
\begin{align}
	T^{\mu_{1}\cdots\mu_{r}}_{\ \qquad\nu_{1}\cdots\nu_{s}}(x(p))            & = \overline{T}^{\rho_{1}\cdots\rho_{r}}_{\ \qquad\sigma_{1}\cdots\sigma_{s}}(y(q))\frac{\partial x^{\mu_{1}}}{\partial y^{\rho_{1}}}\cdots\frac{\partial x^{\mu_{r}}}{\partial y^{\rho_{r}}}\frac{\partial y^{\sigma_{1}}}{\partial x^{\nu_{1}}}\cdots\frac{\partial y^{\sigma_{s}}}{\partial x^{\nu_{s}}} \\
	\overline{T}^{\mu_{1}\cdots\mu_{r}}_{\ \qquad\nu_{1}\cdots\nu_{s}}(y(q)) & = T^{\mu_{1}\cdots\mu_{r}}_{\ \qquad\nu_{1}\cdots\nu_{s}}(x'(p)) \label{duality}.
\end{align}
Equation \eqref{duality} shows the duality between the active viewpoint, i.e, a diffeomorphism, and the passive viewpoint, i.e., a coordinate transformation, on an arbitrary tensor.
We can show
\begin{align}
	\left.\phi^{*}(\nabla_{\bar{\chi}}\overline{T})\right|_{p} = \left.\nabla_{\chi}T\right|_{p}
\end{align}
where $\bar{\chi} \in T_{\phi(p)}N$, $\overline{T}\in T_{p}M^{\otimes r}\otimes T_{p}^{*}M^{\otimes s}$ is an arbitrary $(r,s)$-tensor and we have defined
\begin{align}
	\chi \coloneqq \phi^{*}\bar{\chi} \in T_{p}M,\ T \coloneqq\phi^{*}\overline{T} \in T_{p}M^{\otimes r}\otimes T_{p}^{*}M^{\otimes s}.
\end{align}
Since
\begin{align}
	\phi^{*}(\pounds_{\bar{\chi}}\bar{g})\Big|_{p} = \pounds_{\chi}g\Big|_{p}
\end{align}
holds, we get
\begin{align}
	\phi^{*}(\overline{\nabla}_{\bar{\chi}}\pounds_{\bar{\xi}}\bar{g})\Big|_{p} = \nabla_{\chi}\pounds_{\xi}g\Big|_{p},
\end{align}
where $g = \phi^{*}\bar{g} \in T_{p}^{*}M \otimes T_{p}^{*}M$, $\overline{\nabla}$ and $\nabla$ denote covariant derivatives compatible with $\bar{g}$ and $g$, respectively.
As a consequence, each component satisfies
\begin{align}
	(\pounds_{\bar{\chi}}\bar{g})_{\mu\nu}(y(\phi(p)))                              & = (\pounds_{\chi}g)_{\mu\nu}(x'(p))         \label{duality1}               \\
	(\overline{\nabla}_{\bar{\chi}}\pounds_{\bar{\xi}}\bar{g})_{\mu\nu}(y(\phi(p))) & = (\nabla_{\chi}\pounds_{\xi}g)_{\mu\nu}(x'(p)) \label{duality2}.
\end{align}

\section{The asymptotic behavior of \texorpdfstring{$x'(y)$}{TEXT}}
\label{Flow_proof}
In this appendix, we show that for the algebra $\mathcal{A}$ whose elements satisfying Eq.~\eqref{sufficient1}, Eq.~\eqref{coordinate_asymp} holds.
Let us fix a vector field in $\mathcal{A}$ such that
\begin{align}
    \xi^{\mu}(y) \coloneqq (\mo{1}{},\mo{\rho}{},\mo{1}{},\cdots,\mo{1}{})\quad (\rho \to 0)
    \label{xi}
\end{align}
and consider its integral curve defined by
\begin{align}
\varphi^{\mu}_{\xi;\lambda}(y) \coloneqq \exp[\lambda \xi]y^{\mu}\coloneqq\sum_{n=0}^{\infty}\frac{\lambda^n}{n!}\xi^{n}y^{\mu},
\end{align}
where the action of $\xi^n$ on a function of $y^{\mu}$ is recursively defined as
\begin{align}
    \xi^{n}f(y) &= \xi^{n-1}\xi^{\mu}(y)\partial_{\mu}f(y) \qquad (n=1,2,3,\cdots), \\
    \xi^{0}f(y) &= f(y).
\end{align}
Defining
\begin{align}
\varphi^{\mu}_{\xi;\lambda,n}(y) \coloneqq \frac{\lambda^{n}}{n!}\xi^{n}y^{\mu},
\end{align}
we will show the following proposition:
\begin{proposition}
For any $n\in\mathbb{N}$, it holds
\begin{align}
 \quad \varphi^{\mu}_{\xi;\lambda,n}(y)= (\mo{1}{}, \mo{\rho}{}, \mo{1}{},\cdots, \mo{1}{})\quad (\rho \to 0).\label{eq:asym_phi_n}
\end{align}
\end{proposition}
\begin{proof}
We show the claim by mathematical induction with respect to $n$.
For $n=0$, Eq.~\eqref{eq:asym_phi_n} is clearly satisfied.
Assuming Eq.~\eqref{eq:asym_phi_n} is satisfied for $n=k$, we have
\begin{align}
\varphi^{\mu}_{\xi;\lambda,k+1}(y) &= \frac{\lambda}{k+1}\xi \varphi^{\mu}_{\xi;\lambda,k}(y) \nn \\
&=\frac{\lambda}{k+1}\xi^{\alpha}\partial_{\alpha}(\mo{1}{},\mo{\rho}{},\mo{1}{},\cdots,\mo{1}{})\nn \\
&=(\mo{1}{},\mo{\rho}{},\mo{1}{},\cdots,\mo{1}{}),
\end{align}
where we have used Eq.~\eqref{xi} and the assumption for $n=k$ at the last line.
Therefore, Eq.~\eqref{eq:asym_phi_n} also holds for $n=k+1$, concluding the proof.
\end{proof}

Since the integral curve generated by $\xi^{\mu}$ is given by 
\begin{align}
\varphi^{\mu}_{\xi;\lambda}(y) &= \sum_{n=0}^{\infty}\varphi^\mu_{\xi;\lambda,n}(y),
\end{align}
we have
\begin{align}
    \varphi^{\mu}_{\xi;\lambda}(y) = (\mo{1}{},\mo{\rho}{},\mo{1}{},\cdots,\mo{1}{}).
\end{align}
Next consider the map $\phi^{\mu}_{\xi}(y) \coloneqq \varphi^{\mu}_{\xi;\lambda=1}(y)$.
In general, diffeomorphisms generated by $\mathcal{A}$ and  connected to the identity transformation are given by a product of such maps, i.e.,
\begin{align}
    (\phi_{\xi^{(1)}}\circ \phi_{\xi^{(2)}}\circ\cdots \circ \phi_{\xi^{(N)}})(y)
\end{align}
for some $N$ and vector fields $\xi^{(1)},\xi^{(2)},\cdots,\xi^{(N)}$.
Let us analyze the asymptotic behavior for $N=2$. For two vector fields
\begin{align}
    \left(\xi^{(i)}\right)^{\mu}(y) & =(\mo{1}{},\mo{\rho}{},\mo{1}{},\cdots,\mo{1}{})\quad (i=1,2)
\end{align}
as $\rho\to 0$, we have
\begin{align}
    (\phi_{\xi^{(1)}}\circ \phi_{\xi^{(2)}})^\mu(y)
    =(\mo{1}{},\mo{\rho}{},\mo{1}{},\cdots,\mo{1}{}).
\end{align}
Repeating the same argument, it is shown that the asymptotic behavior of a general diffeomorphism $\phi$ generated by $\mathcal{A}$ is given by
\begin{align}
   \phi^{\mu}(y) = (\mo{1}{},\mo{\rho}{},\mo{1}{},\cdots,\mo{1}{})
\end{align}
for $\rho \to 0$. Therefore, the asymptotic behavior of the corresponding coordinate transformation $x'(y)$ is also given by
\begin{align}
    x'(y) = (\mo{1}{},\mo{\rho}{},\mo{1}{},\cdots,\mo{1}{}).
\end{align}

\section{Supertranslations and superrotations}\label{sec:st_and_sr}

The commutators of vector fields defined in Eqs.~\eqref{eq:xi_st_sr} and \eqref{eq:eta_st_sr} are calculated as
\begin{align}
[\xi, \eta]^{t} &= (V_{1}^{M}\p{M}T_{2} - V_{2}^{M}\p{M}T_{1}) +\mo{\rho}{2},\nn \\
[\xi,\eta]^{\rho} &=\mo{\rho}{2},\nn \\
[\xi,\eta]^{M}&= (V_{1}^{N}\p{N}V_{2}^{M} - V_{2}^{N}\p{N}V_{1}^{M}) + \mo{\rho}{2}
\end{align}
as $\rho\to 0$. 
As a closed algebra including $\xi,\eta$, let us adopt 
\begin{align}
\mathcal{A} \coloneqq \{V=\left(T(x^{M})+\mo{\rho}{2},\mo{\rho}{2}, V^{M}(x^{N}) +\mo{\rho}{2}\right) \mid T, V^{M} \text{ are arbitrary functions of $x^{M}$}\}.
\label{supertranslation_superrotaion_algebra}
\end{align}
From Eqs.~\eqref{omega_Killing_t}-\eqref{omega_Killing_M},
for any $\xi, \eta\in \mathcal{A}$, we have 
\begin{align}
	\omega^{t}(\bar{g}, \pounds_{\eta}\bar{g}, \pounds_{\xi}\bar{g}) = \mo{1}{},\ \omega^{\rho}(\bar{g}, \pounds_{\eta}\bar{g}, 		\pounds_{\xi}\bar{g}) = \mo{\rho}{},\ \omega^{M}(\bar{g}, \pounds_{\eta}\bar{g}, \pounds_{\xi}\bar{g}) = \mo{1}{}
\end{align}
as $\rho\to 0$. Therefore, the corresponding charges are integrable.


\end{document}